\documentclass[journal,12pt,onecolumn,draftclsnofoot]{IEEEtran}

\usepackage{amsmath}
\usepackage{amssymb}
\usepackage{amsfonts}
\usepackage{amsthm}
\usepackage{esint}
\usepackage[hidelinks]{hyperref}
\usepackage{empheq}
\usepackage{tensor} 
\usepackage{scalerel} 
\usepackage{cite}

\usepackage{color}

\newtheorem{theo}{Theorem}
\AfterEndEnvironment{theo}{\noindent\ignorespaces}

\AfterEndEnvironment{cor}{\noindent\ignorespaces}
\newtheorem{lem}{Lemma}
\AfterEndEnvironment{lem}{\noindent\ignorespaces}
\newtheorem{defi}{Definition}
\AfterEndEnvironment{defi}{\noindent\ignorespaces}

\AfterEndEnvironment{conj}{\noindent\ignorespaces}

\AfterEndEnvironment{prop}{\noindent\ignorespaces}
\newtheorem{rem}{Remark}
\AfterEndEnvironment{rem}{\noindent\ignorespaces}

\newcommand{\eqal}[1]{\begin{equation}
\begin{aligned}
#1
\end{aligned}
\end{equation}}

\newcommand{\eqaln}[1]{\begin{equation} \nonumber
\begin{aligned}
#1
\end{aligned}
\end{equation}}

\newcommand{\eqs}[1]{
\begin{equation}
\left\lbrace \begin{aligned}
#1
\end{aligned} \right.
\end{equation}
}

\newcommand{\bos}[1]{\boldsymbol{#1}}

\newcommand{\hapx}{H_{\alpha}(X)_p}

\newcommand{\pr}[3]{H_{\a}\left(#1|#2\right)_{#3}}

\DeclarePairedDelimiterX{\infdivx}[2]{(}{)}{%
  #1\;\delimsize\|\;#2%
}
\newcommand{\petz}{D_\a \infdivx}

\newcommand{\ket}[1]{\vert #1 \rangle}
\newcommand{\bra}[1]{\langle #1 \vert}
\newcommand{\kb}[2]{\left| #1 \vphantom{#2} \right>\left< #2 \vphantom{#1} \right|} 
\newcommand{\proj}[1]{\kb{#1}{#1}} 

\renewcommand{\a}{\alpha}
\newcommand{\sx}{\mathcal{X}}
\newcommand{\sy}{\mathcal{Y}}
\newcommand{\sxm}{\sx \setminus \left\lbrace 1 \right\rbrace}

\newcommand{\gmajj}{\prec_{\scaleto{c}{3pt}}}
\newcommand{\cmajj}{\prec_{\scaleto{\sx}{4pt}}} 
\newcommand{\tr}[2]{\mathrm{Tr}_{#1} \left[ #2 \right]}

\newcommand{\tp}{t}
\newcommand{\tpp}{\tilde{t}}

\begin{document}

\title{A Tight Uniform Continuity Bound for the Arimoto-R\'enyi Conditional Entropy and its Extension to Classical-Quantum States}

\author{Michael~G~Jabbour
        and~Nilanjana~Datta
\thanks{The authors are in the Department of Applied Mathematics and Theoretical Physics, Centre for Mathematical Sciences, University of Cambridge, Cambridge CB3 0WA, United Kingdom}}

\maketitle

\begin{abstract}
  We prove a tight uniform continuity bound for Arimoto's version of the conditional $\alpha$-R\'enyi entropy for the range $\alpha \in [0, 1)$.  This definition of the conditional $\alpha$-R\'enyi entropy is the most natural one among the multiple forms which exist in the literature, since it satisfies two desirable properties of a conditional entropy, namely, the fact that conditioning reduces entropy, and that the associated reduction in uncertainty cannot exceed the information gained by conditioning. Furthermore, it has found interesting applications in various information theoretic tasks such as guessing with side information and sequential decoding. This conditional entropy reduces to the conditional Shannon entropy in the limit $\alpha \to 1$, and this in turn allows us to recover the recently obtained tight uniform continuity bound for the latter from our result. Finally, we apply our result to obtain a tight uniform continuity bound for the conditional $\alpha$-R\'enyi entropy of a classical-quantum state, for $\alpha$ in the same range as above. This again yields the corresponding known bound for the conditional entropy of the state in the limit $\alpha \to 1$.
\end{abstract}


%
\IEEEpeerreviewmaketitle

\section{Introduction}

\IEEEPARstart{I}{n} his seminal paper of 1948~\cite{Shannon}, Claude Shannon introduced the notion of entropy of a discrete random variable $X$ as a measure of its uncertainty or equivalently as a measure of the amount of information we gain on average when we learn its value. The Shannon entropy plays a key role in information theory since it characterises the optimal rate of lossless data compression for a discrete memoryless source. A related entropic quantity or information measure is the conditional Shannon entropy $H(X|Y)$, which quantifies the reduction in the uncertainty of a random variable $X$ when another random variable $Y$ is observed. This and other Shannon information measures (e.g.~the Kullback-Leibler divergence (or relative entropy) and the mutual information) are also of operational significance since they characterise either optimal rates of information-theoretic tasks or fundamental limits of certain statistical inference problems.
Two fundamental properties of a conditional entropy (or equivocation), which are indeed satisfied by the conditional Shannon entropy, are the following:
\begin{enumerate}
\item $H(X|Y) \leq H(X)$;
\item $H(X|Y) \geq H(X) - \log |{\mathcal{Y}}|$, (where~$\mathcal{Y}$ denotes the alphabet of the random variable $Y$).
\end{enumerate}

The above properties are intuitively natural since they ensure that $(i)$ obtaining additional information (in the form of knowledge of a random variable $Y$) can only decrease the uncertainty of a random variable $X$, and $(ii)$ this reduction in uncertainty cannot be more than the amount of information (in bits\footnote{In this paper, we take logarithms to base $2$.}) that has been obtained.
\smallskip

In 1960 Alfred R\'enyi introduced a one-parameter family of entropies~\cite{Renyi1961}, $H_\alpha(X)$ (with $\alpha \in [0,1)\cup (1,\infty)$)  which generalised Shannon's definition of entropy and reduced to it in the limit $\alpha \to 1$. The $\alpha$-R\'enyi entropies are of independent relevance in information theory and also arise naturally in various other branches of mathematics including probability theory, functional analysis, convex geometry and additive combinatorics. In analogy to the Shannon information measures, one can define other information measures related to the $\alpha$-R\'enyi entropy, e.g.~the $\alpha$-R\'enyi divergence and the conditional R\'enyi entropy of order $\alpha$. However, unlike the conditional Shannon entropy, the definition of conditional R\'enyi entropy is not unique. 

The Shannon entropy~\cite{Shannon} extends quite naturally to the conditional case by means of conditional probabilities~\cite{CoverThomas}. In contrast, there are multiple ways to define a notion of conditional R\'enyi entropy, and several definitions have been proposed in the literature~\cite{Arimoto,Cachin,PropertiesRenyi,WorldRenyi,RennerWolfDefCondRenyi,HayashiDefCondRenyi,SkoricObiVerbitskiySchoenmakersDefCondRenyi,FehrBerensCondRenyi,TeixeiraMatosAntunesCondRenyi}. However, some of these definitions are somewhat unsatisfactory since they do not satisfy both of the desired properties $1$ and $2$ mentioned above.
\smallskip

There is, however, one definition of the conditional R\'enyi entropy which respects both of these properties, that is, the $\a$-Arimoto-R\'enyi Conditional Entropy (ARCE) which was introduced by Arimoto~\cite{Arimoto}. Furthermore, it is consistent with the conditional Shannon entropy, in the sense that one recovers the latter from it in the limit $\alpha \to 1$ (a property which does not necessarily hold for the other definitions in the literature). As a consequence, the ARCE has attracted much attention from the information theory community.

The ARCE has been proved to be of operational significance in diverse information-theoretic problems. These include, among others, guessing problems with side information~\cite{Arikan, SasonVerdu2, BracherHofLapidoth,Sundaresan}, sequential decoding~\cite{Arikan}, task encoding with side information~\cite{BunteLapidoth}, list-size capacity of discrete memoryless channels with feedback~\cite{BunteLapidoth2}, Bayesian $M$-ary hypothesis testing~\cite{SasonVerdu}, generalisations of Fano-type inequalities~\cite{Sakai}. See also~\cite{HayashiTan,HayashiTan2}.
\smallskip

Another desirable property of information measures such as entropies and conditional entropies is continuity. A continuity bound quantifies the amount by which the information measure changes when the underlying probability distribution changes by a small amount, and hence provides an estimate of the robustness of the information measure with respect to the probability distribution. For example, given two probability distributions which are at a total variation distance of at most $\epsilon \in (0,1)$, a {\em{uniform continuity bound}} for an entropy is a bound on the difference between their entropies in terms of $\epsilon$ and the alphabet size. It does not depend on the specifics of the probability distributions themselves. The bound is said to be tight if there exists a pair of probability distributions for which it is saturated.

The importance of continuity bounds lies in the fact that we usually lack precise knowledge of the probability distribution and instead only have a rough estimate of it. Hence it is useful to have tight bounds on the error incurred by approximating the probability distribution. Continuity bounds are not only of fundamental interest but also have useful applications in information theory, for example in the study of channel capacities. 

Various such continuity bounds have been derived. A tight uniform continuity bound for the Shannon entropy is attributed to Csisz\'ar, who derived it from Fano's inequality. See also~\cite{Sason2013}. However, it seems to have first appeared in a paper by Zhang~\cite{Zhang}. A tight uniform continuity bound for the conditional Shannon entropy was recently derived by Alhejji and Smith~\cite{AlhejjiSmith}. In the quantum case, Fannes~\cite{Fannes} was the first to prove a uniform continuity bound for the von Neumann entropy, which was later sharpened by Audenaert~\cite{Audenaert}. It was independently obtained by Petz~\cite{PetzCon}. Audenaert's result provides a tight upper bound on the von Neumann entropy of two quantum states of a finite-dimensional system in terms of their trace distance and the dimension of the underlying Hilbert space. Audenaert also derived tight uniform continuity bounds for quantum R\'enyi entropies of order $\a \in [0,1)$ (see Appendix of~\cite{Audenaert}), which covers the case of classical R\'enyi entropies. Continuity bounds for a large family of entropies (which included R\'enyi entropies in the above range) were derived using different proof techniques in~\cite{Eric1}. The authors of~\cite{Eric1} further investigated the case $\alpha > 1$ using the notion of ``majorization flow" in~\cite{Eric2}, obtaining a bound which is sharper than previously known bounds. A similar continuity bound was derived for the conditional von Neumann entropy by Alicki and Fannes~\cite{AlickiFannes}, which in turn was later sharpened by Winter~\cite{Winter2016}. Alhejji and Smith's result for the conditional Shannon entropy was extended to the case of the conditional entropy of a classical-quantum state (with the conditioning being on the classical system) by Wilde~\cite{Wilde2020}.

In this paper we study another important property of the ARCE of order $\alpha$, for the range $\alpha \in [0,1)$. More precisely, we establish a tight uniform continuity bound on the difference of the ARCEs of two joint probability distributions $p_{XY}$ and $q_{XY}$ which are close in total variation distance. See Theorem~\ref{theo:main} of Section~\ref{sec:MainResults} below. To prove this result we introduce $(i)$ the notion of ${\mathcal{X}}$-{\em{majorization}}, which is a special case of the concept of {\em{conditional majorization}} which was introduced by Gour {\em{et al}}\cite{CondMaj} (here ${\mathcal{X}}$ denotes the alphabet of the random variable $X$), and $(ii)$ the notion of a function of a joint probability distribution $p_{XY}$ being {\em{marginally Schur concave}}, that is, Schur concave under ${\mathcal{X}}$-{\em{majorization}}. Ideas and techniques from majorization theory have previously been employed to derive bounds on entropies (see e.g.~\cite{Cicalese2018, Sason2018} and references therein).
\smallskip

We then apply our result to prove a tight uniform continuity bound for the ARCE of order $\alpha$, for $\alpha \in [0,1)$, for classical-quantum (c-q) systems, with the conditioning being done on the classical system. We denote this quantity by $H_\alpha(A|Y)_\rho$, where $Y$ is a classical system ({\em{i.e.}}~ a random variable), $A$ is a finite-dimensional quantum system, and $\rho_{AY}$ is the c-q state of the composite system $AY$ (see \eqref{eq:defi:PetzRenyi} below). See Theorem~\ref{theo:main3} of Section~~\ref{sec:MainResults} in the paper.
\medskip

\noindent
{\bf{Layout of the paper:}} We proceed by first introducing the necessary definitions of the relevant classical and quantum entropies and conditional entropies in Section~\ref{prelim}. We state our main results in Section~\ref{sec:MainResults}, 
and prove them in Section~\ref{sec:proofs}. The key notions of $\sx$-majorization and marginally Schur concave functions, which we exploit in our proofs, are introduced in Section~\ref{sec:ProofIngredients}. We conclude the paper with a discussion and some open problems. Some further tools that we use are stated and proved in the Appendix.

\section{Mathematical preliminaries}\label{prelim}

\subsection{Classical systems}\label{entropies-classical}

Let $\sx \coloneqq \left\lbrace 1,2, \cdots, |\sx| \right\rbrace$, and let $\mathcal{P}_{\sx}$ denote the set of probability distributions on $\sx$. Let $p_X \in \mathcal{P}_{\sx}$. The Shannon entropy of a variable $X$ with distribution $p_X$ is defined as~\cite{Shannon}
\begin{equation}
    H(X)_p \equiv H(p_X) \coloneqq - \sum_{x \in \sx} p_X(x) \log p_X(x).
\end{equation}
The $\alpha$-R\'enyi entropy, for $\alpha \in [0,1) \cup (1,\infty)$ is defined as follows:
\begin{equation}
	\hapx \coloneqq \frac{1}{1-\alpha} \log \left( \sum_{x \in \sx} p_X^{\a}(x) \right).
\end{equation}
The Shannon and $\alpha$-R\'enyi entropies are part of larger set of functions related to the concept of majorization. Given $\bos{u} \in \mathbb{R}^d$ for some dimension $d$, define $\bos{u}^{\downarrow} \in \mathbb{R}^d$ to be the vector containing the elements of $\bos{u}$ arranged in non-increasing order. For $\bos{u}, \bos{v} \in \mathbb{R}^d$, we say $\bos{u}$ is majorized by $\bos{v}$, written $\bos{u} \prec \bos{v}$~\cite{Majorization}, if
\eqal{
  & \sum_{j=1}^k \bos{u}^{\downarrow}_j \leq \sum_{j=1}^k \bos{v}^{\downarrow}_j, \quad \forall k = 1, \cdots, d-1, \\
  \mathrm{and} \quad & \sum_{j=1}^d \bos{u}^{\downarrow}_j = \sum_{j=1}^d \bos{v}^{\downarrow}_j.
}
\smallskip

A function $f:\mathbb{R}^d \to \mathbb{R}$ is said to be Schur-convex if $f(\bos{u}) \leq f(\bos{v})$ for any pair $\bos{u}, \bos{v} \in \mathbb{R}^d$ with $\bos{u} \prec \bos{v}$, and it is said to be Schur-concave if $(-f)$ is Schur-convex. The Shannon and $\alpha$-R\'enyi entropies are notable examples of Schur-concave functions when they are taken as functions of the vectors of probability distributions. An interesting property of majorization is the following: for $\bos{u}, \bos{v} \in \mathbb{R}^d$, $\bos{u} \prec \bos{v}$ if and only if there exists a doubly-stochastic matrix $D$ of dimension $d \times d$ such that $\bos{u} = D \bos{v}$~\cite{HardyLittlewoodPolya}. This property will be useful in our proof.
\medskip

Now let $\sy \coloneqq \left\lbrace 1,2, \cdots, |\sy| \right\rbrace$ and let $\mathcal{P}_{\sx \times \sy}$ denote the set of probability distributions on $\sx \times \sy$. For a pair of random variables $X$ and $Y$ with joint distribution $p_{XY} \in \mathcal{P}_{\sx \times \sy}$, the conditional Shannon entropy $ H(X|Y)_p$ of $X$ conditioned on $Y$ is defined as
\eqal{
  H(X|Y)_p & \coloneqq H(p_{XY}) - H(p_Y) \\
  & = - \sum_{y \in \sy} \sum_{x \in \sx} p_{XY}(x,y) \log p_{X|Y}(x|y),
}
where $p_Y$ denotes the marginal distribution of $Y$.

As mentioned in the Introduction, several definitions of the conditional R\'enyi entropy have been proposed in the literature~\cite{FehrBerensCondRenyi,TeixeiraMatosAntunesCondRenyi}. In this paper we focus on Arimoto's version of the conditional R\'enyi entropy:
\begin{defi} \label{defi:ArimotoRenyi}
For a pair of random variables $X$ and $Y$ with joint distribution $p_{XY} \in \mathcal{P}_{\sx \times \sy}$, the $\a$-Arimoto-R\'enyi conditional entropy, (in short the $\a$-ARCE), for $\alpha \in [0,1) \cup (1,\infty)$ is defined as follows~\cite{Arimoto}:
\begin{equation} \label{eq:ArimotoRenyi2}
	\pr{X}{Y}{p} \coloneqq \frac{\a}{1-\a} \log \left( \sum_{y \in \sy} \left[ \sum_{x \in \sx} p_{XY}^{\a}(x,y) \right]^{1/\a} \right).
\end{equation}
\end{defi}
In the limit $\a \to 1$ it reduces to the conditional Shannon entropy.
\smallskip

Unlike the Shannon and $\alpha$-R\'enyi entropies, the $\a$-ARCE is not Schur-concave. However, it satisfies a weaker notion of Schur concavity which we introduce in Section~\ref{sec:ProofIngredients} and which we exploit in our proofs. This leads us to introduce the concepts of $\mathcal{X}$-{\em{majorization}} and {\em{marginally Schur-concave functions}} (see Section~\ref{sec:ProofIngredients}).
\smallskip

Finally, we recall the definition of the {\em{total variation distance }} (TV) between two probability distributions $p_{XY}, q_{XY} \in \mathcal{P}_{\sx \times \sy}$:
	$\mathrm{TV}(p_{XY},q_{XY}) \coloneqq \frac{1}{2} \sum_{x \in \sx} \sum_{y \in \sy} |p_{XY}(x,y) - q_{XY}(x,y)|.$

\subsection{Classical-quantum systems}\label{entropies-quantum}

In this section we consider finite-dimensional bipartite quantum (and classical-quantum) systems. Their states are given by density matrices, \textit{i.e.}~positive semidefinite operators of unit trace acting on the Hilbert space associated with the system. The trace distance between two states $\rho$ and $\sigma$ of a quantum system is given by $\frac{1}{2}||\rho-\sigma||_1$ where for any operator $A$, acting on the Hilbert space of the system, $||A||_1 := {\mathrm{Tr}}{\sqrt{A^\dagger A}}$.
\medskip

Let us introduce some relevant entropic quantities.
\begin{defi}
Let $\rho_{AB}$ be the density matrix of a bipartite finite-dimensional quantum system $AB$. Then for $\a \in [0, 1) \cup (1, \infty)$ the quantum conditional $\a$-R\'enyi entropy of $A$ given $B$ of the state $\rho_{AB}$ is defined as follows~\cite{Petz}:
\begin{equation} \label{eq:defi:PetzRenyi}
	\pr{A}{B}{\rho} \coloneqq \frac{\a}{1-\a} \log \tr{}{\left( \tr{A}{\rho_{AB}^\a} \right)^{1/\a}}.
\end{equation}
\end{defi}
In the limit $\a \to 1$ it reduces to the quantum conditional entropy 
\begin{align}
    H(A|B)_{\rho} := H({\rho}_{AB}) - H(\rho_B),
\end{align}
where $\rho_B:= \tr{A}{\rho_{AB}}$ denotes the reduced state of the system $B$, and $H(\rho):= - {\mathrm{Tr}}(\rho \log \rho)$ denotes the von Neumann entropy of the state $\rho$.
The quantity $\pr{A}{B}{\rho}$ can also be expressed as follows~\cite{Petz,Tomamichel}:
\begin{equation}
	\pr{A}{B}{\rho} = \sup_{\sigma_B} \left( - \petz{\rho_{AB}}{\mathbb{I}_A \otimes \sigma_B} \right),
\end{equation}
where $\mathbb{I}_A$ denotes the identity operator acting on the system $A$, the supremum is taken over all states $\sigma_B$ of the system $B$, and where for two states $\rho$ and $\sigma$, $\petz{\rho}{\sigma}$ denotes the quantum $\a$-R\'enyi divergence of $\rho$ with respect to $\sigma$. The latter is defined as~\cite{Petz}
\begin{equation}
  \petz{\rho}{\sigma} \coloneqq \frac{1}{\a-1} \log \tr{}{\rho^\a \sigma^{1-\a}},
\end{equation}
if $\mathrm{supp}(\rho) \subseteq \mathrm{supp}(\sigma)$ and $\petz{\rho}{\sigma} \coloneqq + \infty$ else.

We are interested in the conditional $\a$-R\'enyi entropy defined in \eqref{eq:defi:PetzRenyi} since it can be related to the $\a$-Arimoto-R\'enyi conditional entropy when the bipartite system is classical-quantum (c-q) and the conditioning is made on the classical system. Indeed, consider a finite dimensional c-q system with density matrix $\rho_{AY}$, where $A$ denotes a quantum system and $Y$ denotes a classical system of dimension $|\sy|$. We can always write the density matrix as
\begin{equation} \label{eq:rhoAY2}
	\rho_{AY} = \sum_{y \in \sy} r_Y(y) \rho_A^y \otimes \proj{y}_Y,
\end{equation}
where $r_Y \in \mathcal{P}_{\sy}$ and $\left\lbrace \rho_A^y \right\rbrace_{y \in \sy}$ is a set of quantum states of the system $A$. For each $y \in \sy$, consider the spectral decomposition of $\rho_A^y$:
\begin{equation} \label{eq:rhoAY}
	\rho_A^y = \sum_{x \in \sx} r_{X|Y}(x|y) \proj{\phi^y_x}_A,
\end{equation}
where $|\sx| = d_A$ is the dimension of system $A$ and $r_{X|Y}$ is a conditional probability distribution.

We can write the bipartite state as
\begin{equation}
	\rho_{AY} = \sum_{y \in \sy, x \in \sx} r_{XY}(x,y) \proj{\phi^y_x}_A \otimes \proj{y}_Y,
\end{equation}
where we defined $r_{XY}(x,y) = r_Y(y) r_{X|Y}(x|y)$ for all $x \in \sx, y \in \sy$, so that $r_{XY} \in \mathcal{P}_{\sx \times \sy}$. In that case,
\eqal{ \label{eq:PetzToArimoto}
  & \pr{A}{Y}{\rho} \\
  & = \frac{\a}{1-\a} \log \mathrm{Tr} \left[ \left( \mathrm{Tr}_A \left[ \sum_{y \in \sy, x \in \sx} r_{XY}^\a(x,y) \right. \right. \right. \\
  & \hspace{3.5cm} \left. \left. \left. \vphantom{\sum_{y \in \sy, x \in \sx}} \times \proj{\phi^y_x,y}_{AY} \right] \right)^{1/\a} \right] \\
	& = \frac{\a}{1-\a} \log \tr{}{\left( \sum_{y \in \sy, x \in \sx} r_{XY}^\a(x,y) \proj{y}_Y \right)^{1/\a}} \\
	& = \frac{\a}{1-\a} \log \tr{}{ \sum_{y \in \sy} \left( \sum_{x \in \sx} r_{XY}^\a(x,y) \right)^{1/\a} \proj{y}_Y} \\
	& = \frac{\a}{1-\a} \log \left( \sum_{y \in \sy} \left( \sum_{x \in \sx} r_{XY}^\a(x,y) \right)^{1/\a} \right) \\
	& = \pr{X}{Y}{r},
}
where $\ket{\phi^y_x,y}_{AY} \equiv \ket{\phi^y_x}_A \otimes \ket{y}_Y$ and $\pr{X}{Y}{r}$ is the classical $\a$-ARCE for a pair of random variables $X$ and $Y$ with joint probability distribution $r_{XY}$.

In the limit $\a \to 1$, $\pr{A}{Y}{\rho}$ reduces to the conditional entropy of the c-q state $\rho_{AY}$:
\begin{align}\label{q-cond}
 H(A|Y)_{\rho}:= \sum_y r_Y(y) H(\rho_A^y),
\end{align}
where $ H(\rho_A^y)$ denotes the von Neumann entropy of the quantum state $\rho_A^y$.

\section{Main results \label{sec:MainResults}}

In this section we state our main results, namely, tight uniform continuity bounds for the $\a$-Arimoto-R\'enyi conditional entropy ($\a$-ARCE) of a classical joint probability distribution, and the conditional $\a$-R\'enyi entropy of a c-q system. In the limit $\a \to 1$, they reduce to the known bounds for the conditional Shannon entropy and the conditional entropy (of a c-q system). The first theorem pertains to the $\a$-ARCE.

\begin{theo} \label{theo:main}
Let $\a \in [0,1)$, $\epsilon \in (0,1-\frac{1}{|\sx|}]$ and $p_{XY}, q_{XY} \in \mathcal{P}_{\sx \times \sy}$ be such that
\begin{equation} \label{eq:theo:main2}
\mathrm{TV}(p_{XY},q_{XY}) \leq \epsilon.
\end{equation}
Then the following inequality holds:
\eqal{ \label{eq:theo:main}
   & |\pr{X}{Y}{p} - \pr{X}{Y}{q}| \\
   &\leq \frac{1}{1-\a} \log \left( \left( 1- \epsilon \right)^{\a} + \left(|\sx|-1\right)^{1-\a} \epsilon^{\a} \right).
}
Moreover, the inequality is tight, \textit{i.e.}, 
\begin{equation}
\sup_{p_{XY},q_{XY}} \frac{|\pr{X}{Y}{p} - \pr{X}{Y}{q}|}{\gamma(\alpha, \epsilon, |\mathcal{X}|)} =1,
\end{equation}
where $\gamma(\alpha, \epsilon, |\mathcal{X}|)$ denotes the expression on the right hand side of (\ref{eq:theo:main}).
\end{theo}

\begin{rem} \label{rem:ShannonBound}
In the limit $\alpha \to 1$, Theorem~\ref{theo:main} yields the corresponding continuity bound for the conditional Shannon entropy which was derived in~\cite{AlhejjiSmith}, and is given by the following:
for any $\epsilon \in (0,1-\frac{1}{|\sx|}]$, for a pair of probability distributions $p_{XY}$ and $q_{XY}$ in $\mathcal{P}_{\sx \times \sy}$ for which $\mathrm{TV}(q_{XY}, p_{XY}) \leq \epsilon$:
\begin{align}
    |H(X|Y)_p - H(X|Y)_q| \leq \epsilon \log \left(|\sx|-1\right) + h(\epsilon),
\end{align}
where $h(\epsilon):= - \epsilon \log \epsilon - (1-\epsilon) \log (1-\epsilon)$ is the binary entropy.
This can be easily verified by a simple use of l'H\^opital's rule.
\end{rem}

Following the ideas used in~\cite{Wilde2020}, and making use of the bound stated in Theorem~\ref{theo:main}, we derive a tight uniform continuity bound for the conditional $\a$-R\'enyi entropy of a c-q system, with the conditioning being on the classical system, and $\a$ being in the range $[0,1)$. This bound is stated in the following theorem.

\begin{theo} \label{theo:main3}
Consider a c-q system $AY$, where $d_A$ denotes the dimension of the quantum system $A$, and $|\sy|$ denotes the dimension of the classical system $Y$. Let $\rho_{AY}$ and $\sigma_{AY}$ be two states of $AY$ satisfying
\begin{equation} \label{eq:theo:main33}
	\frac{1}{2} || \rho_{AY} - \sigma_{AY} ||_1 \leq \epsilon.
\end{equation}
Then for any $\a \in [0,1)$, and $\epsilon \in (0,1-\frac{1}{d_A}]$ the following inequality holds:
\eqal{ \label{eq:theo:mainQ}
  & |\pr{A}{Y}{\rho} - \pr{A}{Y}{\sigma}| \\
  & \leq \frac{1}{1-\a} \log \left( \left( 1- \epsilon \right)^{\a} + \left(d_A-1\right)^{1-\a} \epsilon^{\a} \right).
}
Moreover, the inequality is tight, \textit{i.e.}, 
\begin{equation}
\sup_{\rho_{AY},\sigma_{AY}} \frac{|\pr{A}{Y}{\rho} - \pr{A}{Y}{\sigma}|}{\gamma(\alpha, \epsilon, d_A)} =1,
\end{equation}
where $\gamma(\alpha, \epsilon, d_A)$ denotes the expression on the right hand side of (\ref{eq:theo:mainQ}).
\end{theo}

\begin{rem}
By making use of Remark~\ref{rem:ShannonBound}, in the limit $\alpha \to 1$ we recover the corresponding continuity bound for the conditional entropy of a c-q state which was derived in~\cite{Wilde2020},
and is given by the following:
for any $\epsilon \in (0,1-\frac{1}{d_A}]$, for a pair of finite-dimensional c-q states $\rho_{AY}, \sigma_{AY}$ satisfying \eqref{eq:theo:main33}, the following inequality holds:
\begin{align}
    |H(A|Y)_{\rho} - H(A|Y)_{\sigma}| \leq \epsilon \log \left(d_A-1\right) + h(\epsilon),
\end{align}
where $h(\epsilon)$ is the binary entropy and $H(A|Y)_{\rho}$ denotes the conditional entropy of the c-q state $\rho_{AY}$.
\end{rem}

\section{Proof of the main results}\label{sec:proofs}

\subsection{Proof ingredients: \texorpdfstring{$\mathcal{X}$}{}-majorization and marginally Schur-concave functions \label{sec:ProofIngredients}}

The notion of conditional majorization was introduced in~\cite{CondMaj}. For completeness, we state its definition in Appendix \ref{CondMajGour}. Here we consider a particular case of conditional majorization, which we refer to as $\mathcal{X}$-\textit{majorization}, as it can be understood as majorization applied to the $\mathcal{X}$-marginals of joint probability distributions $p_{XY} \in \mathcal{P}_{\sx \times \sy}$. The definition is as follows.
\begin{defi}
Denote by $\mathbb{R}^{n \times l}_+$ the set of all $n \times l$ matrices with non-negative values. Consider $P \in \mathbb{R}^{n \times l}_+$ and $Q \in \mathbb{R}^{n \times l}_+$. We say $Q$ is $\mathcal{X}$-majorized by $P$ and write $Q \cmajj P$ if there exist matrices $D^{(j)}$ and $R^{(j)}$, where $j \in \left\lbrace 1, 2, \cdots, l \right\rbrace$, such that
\begin{equation}
	Q = \sum_{j=1}^l D^{(j)} P R^{(j)},
\end{equation}
where each $D^{(j)}$ is an $n \times n$ doubly-stochastic matrix and each $R^{(j)}$ is an $l \times l$ matrix which is such that $\left(R^{(j)}\right)_{ik} = \delta_{ik} \delta_{ij}, \forall i,k \in \left\lbrace 1, 2, \cdots, l \right\rbrace$.
\end{defi}
Henceforth, we choose to represent a joint probability distribution $p_{XY} \in \mathcal{P}_{\sx \times \sy}$ as a matrix of dimensions $|\sx| \times |\sy|$, so that its $y$\textsuperscript{th} column contains the column-vector $\left\lbrace p_{XY}(x,y) \right\rbrace_{x \in \sx}$, for each $y \in \sy$, while its $x$\textsuperscript{th} row contains the row-vector $\left\lbrace p_{XY}(x,y) \right\rbrace_{y \in \sy}$, for each $x \in \sx$. Here we have used the notation $\left\lbrace p_{XY}(x,y) \right\rbrace_{x \in \sx}:=  \left( p_{XY}(1,y), p_{XY}(2,y), \ldots, p_{XY}(|{\mathcal{X}}|,y) \right)$
and $\left\lbrace p_{XY}(x,y) \right\rbrace_{y \in \sy}:=  \left( p_{XY}(x,1), p_{XY}(x,2), \ldots, p_{XY}(x,|{\mathcal{Y}}|) \right)$. We also denote the column-vector of conditional probabilities $p_{X|Y}(x|y)$ for a fixed $y \in \sy$ as
\begin{equation} \label{eq:vecCond}
 	p_{X|Y=y} = \left( p_{X|Y}(1|y), p_{X|Y}(2|y), \ldots, p_{X|Y}(|{\mathcal{X}}||y) \right).
\end{equation}
Moreover, we use the symbols $q_{XY}$ and $p_{XY}$ to refer to both the joint distributions as well as their matrix representations.

We now apply our definition of $\mathcal{X}$-majorization to joint probability distributions in $\mathcal{P}_{\sx \times \sy}$ and relate it to the ``usual" majorization applied to vectors of conditional probabilities. We do so through the following simple lemma.
\begin{lem} \label{lem:cmajToMaj}
Consider $p_{XY}, q_{XY} \in \mathcal{P}_{\sx \times \sy}$. Then the following statements are equivalent:
\begin{enumerate}
	\item $q_{XY} \cmajj p_{XY}$,
	\item $\left( q_{XY} R^{(y)} \bos{e}_{|\sy|} \right) \prec \left( p_{XY} R^{(y)} \bos{e}_{|\sy|} \right)$ for all $y \in \sy$, \label{enum:majCol:lem:cmajToMaj}
	\item $q_{Y} = p_Y$ and $q_{X|Y=y} \prec p_{X|Y=y}$ for all $y \in \sy$.
\end{enumerate}
where $\bos{e}_{|\sy|}$ is a column-vector with $|\sy|$ rows with all its elements equal to $1$.
\end{lem}
Note that point \ref{enum:majCol:lem:cmajToMaj} simply means that each column of $q_{XY}$ is majorized by the corresponding column (same $y$ index) of $p_{XY}$.
\begin{proof}
We first show the equivalence between \textit{1} and \textit{2}. Consider $q_{XY} \cmajj p_{XY}$. In that case, there exist matrices $D^{(z)}$ and $R^{(z)}$, where $z \in \sy$, such that
\begin{equation}
	q_{XY} = \sum_{z=1}^{|\sy|} D^{(z)} p_{XY} R^{(z)},
\end{equation}
where each $D^{(z)}$ is an $|\sx| \times |\sx|$ doubly-stochastic matrix and each $R^{(z)}$ is a $|\sy| \times |\sy|$ matrix which is such that $\left(R^{(z)}\right)_{ab} = \delta_{ab} \delta_{az}, \forall a,b \in \sy$. Using the above equation, for each $y \in \sy$,
\eqal{ \label{eq:eq5}
  q_{XY} R^{(y)} \bos{e}_{|\sy|} & = \sum_{z=1}^{|\sy|} D^{(z)} p_{XY} R^{(z)} R^{(y)} \bos{e}_{|\sy|} \\
  & = D^{(y)} p_{XY} R^{(y)} \bos{e}_{|\sy|},
}
so that~\cite{HardyLittlewoodPolya}
\begin{equation}
	\left( q_{XY} R^{(y)} \bos{e}_{|\sy|} \right) \prec \left( p_{XY} R^{(y)} \bos{e}_{|\sy|} \right), \quad \forall y \in \sy.
\end{equation}
Since the doubly-stochastic matrix $D^{(y)}$ appearing in the last equation of \eqref{eq:eq5} depends on $y$, the proof can be reversed without any loss of generality, so that \textit{1} and \textit{2} are equivalent. We now show the equivalence between \textit{2} and \textit{3}.
Since $p_{XY}(x,y) = p_Y(y) p_{X|Y}(x|y), \forall x \in \sx, \forall y \in \sy$, and similarly for $q_{XY}$, according to the definition in \eqref{eq:vecCond}, we can write, for each $y \in \sy$,
\eqal{
  & p_{Y}(y) p_{X|Y=y} = p_{XY} R^{(y)} \bos{e}_{|\sy|}, \\
  & q_Y(y) q_{X|Y=y} = q_{XY} R^{(y)} \bos{e}_{|\sy|},
}
so that
\eqal{
  & p_{Y}(y) = \sum_{x \in \sx} \left( p_{XY} R^{(y)} \bos{e}_{|\sy|}\right)_x, \qquad \forall y \in \sy, \\
  & q_Y(y) = \sum_{x \in \sx} \left( q_{XY} R^{(y)} \bos{e}_{|\sy|}\right)_x, \qquad \forall y \in \sy.
}
Suppose $\left( q_{XY} R^{(y)} \bos{e}_{|\sy|} \right) \prec \left( p_{XY} R^{(y)} \bos{e}_{|\sy|} \right)$ for all $y \in \sy$. A necessary condition is that $q_Y(y) = p_{Y}(y)$ for all $y \in \sy$, which means that $p_{Y} = q_Y$. It also means that
\begin{equation}
	q_{X|Y=y} \prec p_{X|Y=y}, \quad \forall y \in \sy.
\end{equation}
Again, the proof can be reversed without any loss of generality, so that \textit{2} and \textit{3} are equivalent.
\end{proof}

\noindent We are now in position to define the notion of marginally Schur-concave function.
\begin{defi}
A function $f : \mathbb{R}^{n \times l}_+ \rightarrow \mathbb{R}$ will be called {\em{marginally Schur-concave}} if, for any pair of matrices $P, Q \in \mathbb{R}^{n \times l}_+$, we have
\begin{equation}
	Q \cmajj P \quad \Rightarrow \quad f(P) \leq f(Q).
\end{equation}
The function $f$ is said to be {\em{marginally Schur-convex}} if the opposite inequality holds.
\end{defi}

We show the following.
\begin{lem} \label{lem:CondSchurConcavity}
The $\a$-ARCE is marginally Schur-concave for all $\a \in [0,1) \cup (1,\infty)$.
\end{lem}
\begin{proof}
The $\a$-ARCE can be written as
\begin{equation} \label{eq:ArimotoRenyi3}
	\pr{X}{Y}{p} = \frac{\a}{1-\a} \log \left( \sum_{y \in \sy} p_Y(y) ||p_{X|Y=y}||_\a \right).
\end{equation}
We need to show that $q_{XY} \cmajj p_{XY}$ implies $\pr{X}{Y}{p} \leq \pr{X}{Y}{q}$. From Lemma \ref{lem:cmajToMaj}, we have that $q_{XY} \cmajj p_{XY}$ implies $p_{Y} = q_Y$ and $q_{X|Y=y} \prec p_{X|Y=y}$, for all $y \in \sy$. We distinguish between the two cases $\a>1$ and $\a \in [0,1)$.
\paragraph{Case 1: $\a>1$.}
In this case, the function $u \mapsto u^\a$ is convex for $u \in [0, \infty)$, so that the $\a$-norm $||p_{X|Y=y}||_\a$ is Schur-convex. Since $q_{X|Y=y} \prec p_{X|Y=y}$ for all $y \in \sy$, this implies that
\begin{equation}
	||p_{X|Y=y}||_\a \geq ||q_{X|Y=y}||_\a, \quad \forall y \in \sy.
\end{equation}
Since $p_{Y} = q_Y$ and $1-\a<0$, we end up with
\begin{equation}
	\pr{X}{Y}{p} \leq \pr{X}{Y}{q}, \quad \forall \a > 1.
\end{equation}
\paragraph{Case 2: $\a \in [0,1)$.}
In this case, the function $u \mapsto u^\a$ is concave for $u \in [0, \infty)$, so that the $\a$-quasi-norm $||p_{X|Y=y}||_\a$ is Schur-concave. Since $q_{X|Y=y} \prec p_{X|Y=y}$ for all $y \in \sy$, this implies that
\begin{equation}
	||p_{X|Y=y}||_\a \leq ||q_{X|Y=y}||_\a, \quad \forall y \in \sy.
\end{equation}
Since $p_{Y} = q_Y$ and $1-\a>0$, we end up with
\begin{equation}
	\pr{X}{Y}{p} \leq \pr{X}{Y}{q}, \quad \forall \a \in [0,1).
\end{equation}
\end{proof}

\noindent Note that Lemma \ref{lem:CondSchurConcavity} can also be proven by exploiting the fact that $\mathcal{X}$-majorization is a special case of conditional majorization, and by making use of Theorem 1 of~\cite{CondMaj}, which we state in the Appendix \ref{CondMajGour} as Lemma \ref{lem:CharacCondMajGour}.
\begin{rem}
Using the same ideas as in the proof of Lemma \ref{lem:CondSchurConcavity}, the conditional Shannon entropy can be readily verified to be marginally Schur-concave (see also~\cite{AlhejjiSmith}).
\end{rem}

\subsection{Proof of Theorem~\ref{theo:main}}

Fix $\epsilon \in [0,1-\frac{1}{|\sx|})$.
Consider two probability distributions $p_{XY}$ and $q_{XY}$ such that their total variation distance $\tp \coloneqq \mathrm{TV}(p_{XY},q_{XY})$ is at most equal to $\epsilon$. We assume, without loss of generality, that $\pr{X}{Y}{p} \geq \pr{X}{Y}{q}$ and define $\Delta H_\a \coloneqq \pr{X}{Y}{p} - \pr{X}{Y}{q}$.

The proof of Theorem~\ref{theo:main} consists of a series of steps which are described in detail below. The key idea is to alter the distributions $q_{XY}$ and $p_{XY}$, in a series of iterative steps such that the difference $\pr{X}{Y}{p}-\pr{X}{Y}{q}$ never decreases, and the total variation distance also remains unchanged. As a result of these manipulations, $\pr{X}{Y}{q}$ reduces to zero while $\pr{X}{Y}{p}$ attains its maximal value under the constraint on the total variation distance, which indeed is the desired upper bound. Note that our initial steps (Step A and Step B) are similar to those used by Alhejji and Smith~\cite{AlhejjiSmith}. However, the remaining steps deviate considerably from theirs in order to deal with the additional complexity of the $\a$-ARCE (compared to the conditional Shannon entropy). In order to clarify some of the steps of the proof, we consider a $3 \times 3$ example in Appendix~\ref{exampleARCEproof}.
\medskip

\noindent
{\bf{Step A: Reordering}} Recall that we chose to represent the probability distributions $p_{XY}$ and $q_{XY}$ by $|\sx|\times |\sy|$ matrices. For convenience, we first arrange the $|\sy|$ columns of $q_{XY}$ such that 
\begin{equation}
 q_Y(y) \geq q_Y(y+1), \quad \forall y \in \sy \setminus \left\lbrace |\sy| \right\rbrace.
\end{equation}
We then do the same permutation of columns of $p_{XY}$. Note, however, that $p_Y(y)$ is not necessarily greater that $p_Y(y+1)$.
From the form given in \eqref{eq:ArimotoRenyi3}, one understands that this does not change the value of the $\a$-ARCE. Now, for the $y$\textsuperscript{th} column, define the two sets
\eqal{ \label{eq:eq6}
	I_y & = \left\lbrace x | q_{XY}(x,y) \geq p_{XY}(x,y) \right\rbrace, \\
	I_y^c & = \left\lbrace x | q_{XY}(x,y) < p_{XY}(x,y) \right\rbrace.
}
Do the following changes for both $p_{XY}$ and $q_{XY}$. Within the $y$\textsuperscript{th} column, if both $I_y$ and $I_y^c$ are non-empty, put the elements corresponding to indices in $I_y$ ahead of those corresponding to indices in $I_y^c$. Finally, in both $p_{XY}$ and $q_{XY}$, order the elements corresponding to indices in $I_y$ such that
\begin{equation}
	q_{XY}(x+1,y) \leq q_{XY}(x,y), \quad \forall x \in I_y \setminus \left\lbrace |I_y| \right\rbrace,
\end{equation}
and do the same for the elements corresponding to indices in $I_y^c$, so that we also have
\begin{equation}
	q_{XY}(x+1,y) \leq q_{XY}(x,y), \quad \forall x \in I_y^c \setminus \left\lbrace |I_y^c| \right\rbrace.
\end{equation}
Note that it \textit{does not} necessarily mean that
\begin{equation}
	p_{XY}(x+1,y) \leq p_{XY}(x,y), \quad \forall x \in I_y \setminus \left\lbrace |I_y| \right\rbrace,
\end{equation}
or,
\begin{equation}
	p_{XY}(x+1,y) \leq p_{XY}(x,y), \quad \forall x \in I_y^c \setminus \left\lbrace |I_y^c| \right\rbrace.
\end{equation}
These operations correspond to permutations of elements within fixed columns of $p_{XY}$ and $q_{XY}$, in the sense that no element is transferred from one column to a different column. For an illustration, see Figure~\ref{exampleARCEproofA} of Appendix~\ref{exampleARCEproof}.

As a consequence, neither $\pr{X}{Y}{p}$ nor $\pr{X}{Y}{q}$ changes under these operations. Furthermore, the total variation distance between the probability distributions does not change either, since the exact same permutations have been performed in both $p_{XY}$ and $q_{XY}$.
\medskip

\noindent
{\bf{Step B: Walking}} We start by moving probability weights in $q_{XY}$. Specifically, we will make $q_{XY}$ less disordered in the sense of $\mathcal{X}$-majorization, by ``concentrating" probability weights, which will decrease the value of $\pr{X}{Y}{q}$ in the process. For all $y \in \sy$ such that $I_y$ is not empty, define the column vectors $\bos{v}^{(0)}(y):= \left\lbrace q_{XY}(x,y)\right\rbrace_{x \in \sx}$ and $\bos{v}^{(i)}(y)\equiv \left\lbrace v_x^{(i)}(y) \right\rbrace_{x \in \sx}$ for $i = 1, \cdots, |I_y|-1$ such that
\eqal{ \label{eq:eq4}
	& \bos{v}_1^{(i)}(y) = \bos{v}_1^{(i-1)}(y) + \left[ q_{XY}(i+1,y) - p_{XY}(i+1,y) \right], \\
	& \bos{v}_{i+1}^{(i)}(y) = p_{XY}(i+1,y), \\
	& \bos{v}_x^{(i)}(y) = \bos{v}_x^{(i-1)}(y), \quad \forall x \in \sx \setminus \left\lbrace 1, i+1 \right\rbrace.
}
Similarly, for all $y \in \sy$ such that $I_y$ is empty, define the column vectors $\bos{u}^{(0)}(y):= \left\lbrace 
q_{XY}(x,y)\right\rbrace_{x \in {\mathcal{X}}}$ and $\bos{u}^{(i)}(y):= \left\lbrace u_x^{(i)}(y) \right\rbrace_{x \in \sx}$ for $i = 1, \cdots, |\sx|-1$ such that
\eqal{ \label{eq:eq4Comp}
	& \bos{u}_1^{(i)}(y) = \bos{u}_1^{(i-1)}(y) + \delta_i, \\
	& \bos{u}_{i+1}^{(i)}(y) = \bos{u}_{i+1}^{(i-1)}(y) - \delta_i, \\
	& \bos{u}_x^{(i)}(y) = \bos{u}_x^{(i-1)}(y), \quad \forall x \in \sx \setminus \left\lbrace 1, i+1 \right\rbrace.
}
where $\delta_i = \min \left\lbrace \bos{u}_{i+1}^{(i-1)}(y), p_{XY}(1,y)-\bos{u}_{1}^{(i-1)}(y) \right\rbrace$. The above equations model the following moves in the matrix representing $q_{XY}$: in columns for which $I_y$ is empty, at each step $i = 1, \cdots, |\sx|-1$, we remove a probability weight $\delta_i$ from $q_{XY}(i+1,y)$ an add it into $q_{XY}(1,y)$. We do so as long as the inequality $q_{XY}(1,y) \leq p_{XY}(1,y)$ is not violated. This is in order to keep the total variation distance unchanged. Indeed, note that the equations imply that we stop adding weight into $q_{XY}(1,y)$ either when we reach $q_{XY}(1,y) = p_{XY}(1,y)$, or when $q_{XY}(i,y) = 0$ for all $i = 2, \cdots, |\sx|$. For an illustration, see Figure~\ref{exampleARCEproofB} of Appendix~\ref{exampleARCEproof}.

From the first identity in \eqref{eq:eq6} and Lemmas \ref{lem:majEps} and \ref{lem:MajPerp},
\begin{equation}
	\bos{v}^{(i)}(y) \prec \bos{v}^{(i+1)}(y),
\end{equation}
at each step $i = 1, \cdots, |I_y|-1$, for all $y \in \sy$, so that
\begin{equation}
	\bos{v}^{(0)}(y) \prec \bos{v}^{(|I_y|-1)}(y),
\end{equation}
for all $y \in \sy$.

Similarly, we have that
\begin{equation}
	\bos{u}^{(0)}(y) \prec \bos{u}^{(|\sx|-1)}(y).
\end{equation}
Define the matrix $\tilde{q}_{XY} \in \mathcal{P}_{\sx \times \sy}$ whose columns are given by the final column-vectors $\bos{v}^{(|I_y|-1)}(y)$ when $y$ is such that $I_y$ is non-empty, and the column-vectors $\bos{u}^{(|\sx|-1)}(y)$ else, \textit{i.e.},
\begin{equation}
    \tilde{q}_{XY}(x,y) = \bos{v}_x^{(|I_y|-1)}(y),  \quad \forall x \in \sx,
\end{equation}
for all $y$ such that $I_y$ is non-empty, and
\begin{equation}
    \tilde{q}_{XY}(x,y) = \bos{u}^{(|\sx|-1)}(y),  \quad \forall x \in \sx,
\end{equation}
for all $y$ such that $I_y$ is empty. We therefore have that the columns of $q_{XY}$ are majorized by the corresponding columns of $\tilde{q}_{XY}$, \textit{i.e.},
\begin{equation}
	\left( q_{XY} R^{(y)} \bos{e}_{|\sy|} \right) \prec \left( \tilde{q}_{XY} R^{(y)} \bos{e}_{|\sy|} \right), \quad \forall y \in \sy.
\end{equation}
According to Lemma \ref{lem:cmajToMaj}, it means that $q_{XY} \cmajj \tilde{q}_{XY}$, which, using Lemma \ref{lem:CondSchurConcavity}, implies that
\begin{equation}
\pr{X}{Y}{\tilde{q}} \leq \pr{X}{Y}{q}.
\end{equation}
In other words, $\pr{X}{Y}{q}$ does not increase under the above operations. For the sake of clarity, we relabel $\tilde{q}_{XY}$ as $q_{XY}$ at this point. We therefore now have
\begin{equation}
	\Delta H_\a \leq \pr{X}{Y}{p}-\pr{X}{Y}{q}.
\end{equation}
Define the two sets
\eqal{
	J & = \left\lbrace y | q_{XY}(1,y) \geq p_{XY}(1,y) \right\rbrace, \\
	J^c & = \left\lbrace y | q_{XY}(1,y) < p_{XY}(1,y) \right\rbrace,
}
(where $J^c$ might be empty) and put the columns of $J$ ahead of the columns of $J^c$ in both $p_{XY}$ and $q_{XY}$. After all these replacements, our joint probability matrices satisfy the following relations:
\eqal{ \label{eq:condStepB.1.3}
	q_{XY}(1,y) - p_{XY}(1,y) \geq 0, & \\
	q_{XY}(x,y) - p_{XY}(x,y) \leq 0, & \quad \forall \ x \in \sxm,
}
for all $y \in J$ and
\eqal{ \label{eq:condStepB.1.4}
	q_{XY}(1,y) - p_{XY}(1,y) < 0, & \\
	q_{XY}(x,y) = 0, & \quad \forall \ x \in \sxm.
}
for all $y \in J^c$. Note that this implies
\begin{equation} \label{eq:EpsStepB.1}
	\sum_{y \in J} \left( q_{XY}(1,y) - p_{XY}(1,y) \right) = \tp \equiv {\mathrm{TV}}(p_{XY}, q_{XY}).
\end{equation}
\medskip

\noindent
{\bf{Step C: Enlarging}} We will now apply a complementary operation to $p_{XY}$, this time making the probabilities more ``spread out" within each column, making the probability distribution more disordered in the sense of $\mathcal{X}$-majorization, which will now increase the value of $\pr{X}{Y}{p}$. In order to do this, we increase the dimensions of (the matrix representing) $p_{XY}$. The extra dimensions which are added will only be exploited in the context of the proof, and will not influence the result. Note that since we increase the dimension of $p_{XY}$, we also do the same for $q_{XY}$ by just appending zeros in the extra dimensions. We do this in order to always compare matrices of similar dimensions. Denote by $\mathcal{P}'$ the space of matrices of dimensions $(2|\sx|-1) \times |\sy|$, and define $q'_{XY}$ and $p'_{XY}$ in $\mathcal{P}'$ such that
\begin{equation} \label{eq:condStepB.2.2}
  q'_{XY}(x,y) = \begin{cases}
    q_{XY}(x,y), & \forall \ x = 1, \cdots, |\sx|, \\
    0, & \forall \ x = |\sx| + 1, \cdots, 2|\sx|-1,
  \end{cases}
\end{equation}
for all $y \in \sy$, and similarly,
\begin{equation} \label{eq:condStepB.2.3}
  p'_{XY}(x,y) = \begin{cases}
    p_{XY}(x,y), & \forall \ x = 1, \cdots, |\sx|, \\
    0, & \forall \ x = |\sx| + 1, \cdots, 2|\sx|-1,
  \end{cases}
\end{equation}
for all $y \in \sy$. At this point, $\Delta H_\a$ is upper-bounded by $\pr{X}{Y}{p'}-\pr{X}{Y}{q'}$. Next, we make $p'_{XY}$ more disordered, therefore increasing the value of $\pr{X}{Y}{p'}$, without changing $q'_{XY}$ and without changing the total variation distance between the two matrices. Construct the matrix $p''_{XY}$ in $\mathcal{P}'$ such that
\begin{equation} \label{eq:condStepB.2.1.1}
  p''_{XY}(1,y) = p'_{XY}(1,y),
\end{equation}
\begin{equation} \label{eq:condStepB.2.1.2}
  p''_{XY}(x,y) = q'_{XY}(x,y),
\end{equation}
for all $x = 2, \cdots, |\sx|$ and
\begin{equation} \label{eq:condStepB.2.1.3}
  p''_{XY}(x,y) = p'_{XY}(x+1-|\sx|,y) - q'_{XY}(x+1-|\sx|,y),
\end{equation}
for all $x = |\sx| + 1, \cdots, 2|\sx|-1$, for all $y \in \sy$. For an illustration, see Figure~\ref{exampleARCEproofC} of Appendix~\ref{exampleARCEproof}.

Using Lemmas \ref{lem:MajPerp} and \ref{lem:majEpsComp}, we have that each column of $p''_{XY}$ is majorized by the corresponding column of $p'_{XY}$. In other words, using the notations of $\mathcal{X}$-majorization, 
\begin{equation}
	\left( p''_{XY} R^{(y)} \bos{e}_{|\sy|} \right) \prec \left( p'_{XY} R^{(y)} \bos{e}_{|\sy|} \right), \quad \forall \ y \in \sy,
\end{equation}
so that according to Lemma \ref{lem:cmajToMaj}, $p''_{XY} \cmajj p'_{XY}$, and using Lemma \ref{lem:CondSchurConcavity}, $\pr{X}{Y}{p'} \leq \pr{X}{Y}{p''}$. We have at this point
\begin{equation} \label{eq:DeltaHa}
	\Delta H_\a \leq \pr{X}{Y}{p''}-\pr{X}{Y}{q'}.
\end{equation}
Furthermore, we have
\begin{equation} \label{eq:epsPrime}
	\tpp \coloneqq \sum_{y=1}^{|\sy|} \sum_{x=|\sx|+1}^{2|\sx|-1} p''_{XY}(x,y) \leq \tp.
\end{equation}
In fact, $\tpp = \tp$ if $J^c$ is empty. Equivalently,
\begin{equation} \label{eq:epsPrime2}
	\tpp= \tp - \sum_{y \in J^c} (p'_{XY}(1,y)-q'_{XY}(1,y)).
\end{equation}
\medskip


\noindent
{\bf{Step D: Bounding the terms individually}}
\smallskip

{\bf{Step D.1: Upper bounding $\pr{X}{Y}{p''}$}} Now, we focus on $p''_{XY}$. We have
\eqaln{
  \pr{X}{Y}{p''} & = \frac{\a}{1-\a} \log \left( \sum_{y \in \sy} \left[ \sum_{x=1}^{2|\sx|-1} (p''_{XY}(x,y))^{\a} \right]^{1/\a} \right) \\
	& = \frac{1}{1-\a} \log \left( \left\lVert \sum_{x=1}^{2|\sx|-1} (p''_{XY}(x, \cdot))^{\a} \right\rVert_{1/\a} \right) \\
  & = \frac{1}{1-\a} \log \left( \left\lVert \sum_{x=1}^{|\sx|} (p''_{XY}(x, \cdot))^{\a} \right. \right. \\
  & \hspace{2.2cm} \left. \left. + \sum_{x=|\sx|+1}^{2|\sx|-1} (p''_{XY}(x, \cdot))^{\a} \right\rVert_{1/\a} \right) \\
	& \leq \frac{1}{1-\a} \log \left( \left\lVert \sum_{x=1}^{|\sx|} (p''_{XY}(x, \cdot))^{\a} \right\rVert_{1/\a} \right. \\
  & \hspace{2.2cm} \left. + \left\lVert \sum_{x=|\sx|+1}^{2|\sx|-1} (p''_{XY}(x, \cdot))^{\a} \right\rVert_{1/\a}\right) \\
	& \leq \frac{1}{1-\a} \log \left( \left\lVert \sum_{x=1}^{|\sx|} (p''_{XY}(x, \cdot))^{\a} \right\rVert_{1/\a} \right. \\
  & \hspace{2.2cm} \left. + \sum_{x=|\sx|+1}^{2|\sx|-1} \left\lVert (p''_{XY}(x, \cdot))^{\a} \right\rVert_{1/\a}\right)
}
where the inequalities follow from the Minkowski inequality for $1/\a > 1$, and the monotonicity of the logarithm. The second term inside the $\log$ on the right-hand side of the last inequality is
\begin{equation*}
	\sum_{x=|\sx|+1}^{2|\sx|-1} \left\lVert (p''_{XY}(x, \cdot))^{\a} \right\rVert_{1/\a} = \sum_{x=|\sx|+1}^{2|\sx|-1} \left( \sum_{y \in \sy} p''_{XY}(x,y) \right)^{\a}.
\end{equation*}
The elements of $p''_{XY}$ involved in the above equation satisfy \eqref{eq:epsPrime}, so that
\eqaln{
  & \left( \frac{\tpp}{|\sx|-1}, \cdots, \frac{\tpp}{|\sx|-1} \right) \\
  & \prec \left( \sum_{y \in \sy} p''_{XY}(|\sx|+1,y), \cdots, \sum_{y \in \sy} p''_{XY}(2|\sx|-1,y) \right),
}
where both the vectors involved in the above equation contain $|\sx|-1$ elements. Since the $\a$-quasi-norm is Schur-concave for $\a < 1$ (see proof of Lemma \ref{lem:CondSchurConcavity}), we get that
\eqal{
  & \left\lVert \left( \sum_{y \in \sy} p''_{XY}(|\sx|+1,y), \cdots, \sum_{y \in \sy} p''_{XY}(2|\sx|-1,y) \right) \right\rVert_{\a} \\
  & \leq \left\lVert \left( \frac{\tpp}{|\sx|-1}, \cdots, \frac{\tpp}{|\sx|-1} \right) \right\rVert_{\a},
}
or,
\begin{equation}
	\sum_{x=|\sx|+1}^{2|\sx|-1} \left( \sum_{y \in \sy} p''_{XY}(x,y) \right)^{\a} \leq \left( |\sx|-1 \right)^{1-\a} \tpp^\a,
\end{equation}
which leads to
\eqaln{
  & \pr{X}{Y}{p''} \\
  & \leq \frac{1}{1-\a} \log \left( \left\lVert \sum_{x=1}^{|\sx|} (p''_{XY}(x, \cdot))^{\a} \right\rVert_{1/\a} + \left( |\sx|-1 \right)^{1-\a} \tpp^\a \right).
}

At this point, we define $r_{XY} \in \mathcal{P}$ such that
\begin{equation}
	r_{XY}(x,y) = p''_{XY}(x,y), \quad \forall \ y \in \sy, \quad \forall \ x \in \sx,
\end{equation}
or, using \eqref{eq:condStepB.2.1.1}, \eqref{eq:condStepB.2.1.2} and \eqref{eq:condStepB.2.1.3},
\eqal{ \label{eq:condStepB.3}
	r_{XY}(1,y) = p'_{XY}(1,y), & \\
	r_{XY}(x,y) = q'_{XY}(x,y), & \quad \forall \ x = 2, \cdots, |\sx|,
}
for all $y \in \sy$, so that
\eqal{ \label{eq:Hpsec}
  & \pr{X}{Y}{p''} \\
  & \leq \frac{1}{1-\a} \log \left( \left\lVert \sum_{x=1}^{|\sx|} (r_{XY}(x, \cdot))^{\a} \right\rVert_{1/\a} + \left( |\sx|-1 \right)^{1-\a} \tpp^\a \right).
}
Note that
\begin{equation} \label{eq:epsPrime3}
	\sum_{y=1}^{|\sy|} \sum_{x=1}^{|\sx|} r_{XY}(x,y) = 1-\tpp \geq 1-\tp.
\end{equation}
\medskip

{\bf{Step D.2: Lower bounding $\pr{X}{Y}{q'}$}} We now turn to $q'_{XY}$. We have
\eqaln{
	\pr{X}{Y}{q'} & = \frac{\a}{1-\a} \log \left( \sum_{y \in \sy} \left[ \sum_{x=1}^{|\sx|} (q'_{XY}(x,y))^{\a} \right]^{1/\a} \right) \\
  & = \frac{\a}{1-\a} \log \left( \sum_{y \in J} \left[ \sum_{x=1}^{|\sx|} (q'_{XY}(x,y))^{\a} \right]^{1/\a} \right. \\
  & \hspace{2.2cm} \left. + \sum_{y \in J^c} \left[ \sum_{x=1}^{|\sx|} (q'_{XY}(x,y))^{\a} \right]^{1/\a} \right) \\
	& = \frac{\a}{1-\a} \log \left( \sum_{y \in J} \left[ \sum_{x=1}^{|\sx|} (q'_{XY}(x,y))^{\a} \right]^{1/\a} \right. \\
  & \hspace{2.2cm} \left. + \sum_{y \in J^c} q'_{XY}(1,y) \right)
}
where we made use of \eqref{eq:condStepB.1.4}. Define $q''_{XY} \in \mathcal{P}$ such that
\eqal{
	q''_{XY}(x,y) = q'_{XY}(x,y)-r_{XY}(x,y), & \quad \forall \ y \in J, \quad x = 1, \\
	q''_{XY}(x,y) = 0, & \quad \mathrm{else}.
}
The elements of $q''_{XY}$ are all non-negative as a consequence of \eqref{eq:condStepB.3}, \eqref{eq:condStepB.2.2}, \eqref{eq:condStepB.2.3} and \eqref{eq:condStepB.1.3}. Furthermore, using \eqref{eq:EpsStepB.1}, we have that
\begin{equation} \label{eq:EpsStepB.4}
	\sum_{y \in J} q''_{XY}(1,y) = \tp.
\end{equation}
Using this, we can write
\eqal{
	\pr{X}{Y}{q'} & = \frac{\a}{1-\a} \log \left( \sum_{y \in J} \left\lVert r_{XY}(.,y) + q''_{XY}(.,y) \right\rVert_{\a} \right. \\
  & \hspace{2.2cm} \left. + \sum_{y \in J^c} q'_{XY}(1,y) \right)
}
and using the reverse Minkowski inequality for $\a < 1$, and the monotonicity of the logarithm, we end up with
\eqaln{
  & \exp \left( \frac{1-\a}{\a} \pr{X}{Y}{q'} \right) \\
  & \geq \sum_{y \in J} \left\lVert r_{XY}(.,y) \right\rVert_{\a} + \sum_{y \in J} \left\lVert q''_{XY}(.,y) \right\rVert_{\a} + \sum_{y \in J^c} q'_{XY}(1,y) \\
	& = \sum_{y \in J} \left\lVert r_{XY}(.,y) \right\rVert_{\a} + \sum_{y \in J} q''_{XY}(1,y) + \sum_{y \in J^c} q'_{XY}(1,y) \\
	& = \sum_{y \in \sy} \left\lVert r_{XY}(.,y) \right\rVert_{\a} - \sum_{y \in J^c} r_{XY}(1,y) + \tp + \sum_{y \in J^c} q'_{XY}(1,y)
}
using \eqref{eq:EpsStepB.4}. Now, from \eqref{eq:condStepB.3} and \eqref{eq:epsPrime2}, we obtain
\begin{equation}
	\tp + \sum_{y \in J^c} q'_{XY}(1,y) - \sum_{y \in J^c} r_{XY}(1,y) = \tpp,
\end{equation}
so that
\begin{equation} \label{eq:Hqpri}
	\pr{X}{Y}{q'} \geq \frac{\a}{1-\a} \log \left( \sum_{y \in \sy} \left\lVert r_{XY}(.,y) \right\rVert_{\a} + \tpp \right).
\end{equation}

{\bf{Step E: Upper bounding the difference}} Putting \eqref{eq:DeltaHa}, \eqref{eq:Hpsec} and \eqref{eq:Hqpri} together, we see that our initial conditional entropy difference $\Delta H_\a$ is upper bounded as follows:
\eqaln{
\Delta H_\a & \leq \frac{1}{1-\a} \left[ \log \left( R(r_{XY})^\a + \left( |\sx|-1 \right)^{1-\a} \tpp^\a \right) \right. \\
& \hspace{1.5cm} \left. - \log \left( \vphantom{R(r_{XY})^\a + \left( |\sx|-1 \right)^{1-\a} \tpp^\a} \left[ R(r_{XY}) + \tpp\right]^\a \right) \right],
}
where we defined $R(r_{XY}) = \sum_{y \in \sy} \left\lVert r_{XY}(.,y) \right\rVert_{\a}$. Since $f(u) = u^\a$ is concave for $u \geq 0$, $\a \in [0,1)$, and $f(0) = 0$, it is also subadditive, so that
\eqaln{
	R(r_{XY}) & = \sum_{y \in \sy} \left( \sum_{x \in \sx} r_{XY}^{\a}(x,y) \right)^{1/\a} \\
	& \geq \sum_{y \in \sy} \sum_{x \in \sx} r_{XY}(x,y) \\
	& = 1 - \tpp
}
where we used \eqref{eq:epsPrime3}.

Furthermore, $\tpp \leq \tp \leq \epsilon \leq 1-1/|\sx|$. Using this along with Lemma \ref{lem:IncDiffLog}, we have that
\eqaln{
	\Delta H_\a & \leq \frac{1}{1-\a} \left[ \log \left( \left( 1 - \tpp \right)^\a + \left( |\sx|-1 \right)^{1-\a} \tpp^\a \right) \right. \\
  & \hspace{1.5cm} \left.  - \log \left( \left[ 1 - \tpp + \tpp\right]^\a \right) \right] \\
	& = \frac{1}{1-\a} \log \left( \left( 1 - \tpp \right)^\a + \left( |\sx|-1 \right)^{1-\a} \tpp^\a \right).
}
Finally, using Lemma \ref{lem:BoundMonotonic} along with the fact that $\tpp \leq \epsilon$, we conclude that 
\begin{equation}
	\Delta H_\a \leq \frac{1}{1-\a} \log \left( \left( 1 - \epsilon \right)^\a + \left( |\sx|-1 \right)^{1-\a} \epsilon^\a \right),
\end{equation}
which ends the proof of the bound. To see that the inequality is tight, one can consider the probability distributions whose elements satisfy the following relations:
\eqal{ \label{eq:qtight}
	& q_{XY}(1,1) = 1, \\
	& q_{XY}(x,1) = 0, \quad \forall \ x \in \sxm, \\
	& q_{XY}(x,y) = 0, \quad \forall \ x \in \sx, \forall \ y \in \sy \setminus \left\lbrace 1 \right\rbrace,
}
and
\eqal{ \label{eq:ptight}
	& p_{XY}(1,1) = 1-\epsilon \\
	& p_{XY}(x,1) = \frac{\epsilon}{|\sx|-1}, \quad \forall \ x \in \sxm, \\
	& p_{XY}(x,y) = 0, \quad \forall \ x \in \sx, \forall \ y \in \sy \setminus \left\lbrace 1 \right\rbrace.
}

\subsection{Proof of Theorem~\ref{theo:main3}}

Our proof is analogous to the proof of the uniform continuity bound for the conditional entropy of c-q states by Wilde~\cite{Wilde2020}. It relies on the use of a conditional dephasing channel and the data processing inequality to go from a c-q setting to a classical-classical one. However, the unitality of the dephasing channel, which was exploited by Wilde in his proof, cannot be used in the case of the conditional $\a$-R\'enyi entropy. Instead we use the notion of $\sx$-majorization. This allows us to then employ our Theorem~\ref{theo:main} to arrive at the desired result. 

Consider the following decompositions for $\rho_{AY}$ and $\sigma_{AY}$:
\begin{equation} \label{eq:theo:main31}
	\rho_{AY} = \sum_{y \in \sy} r_Y(y) \rho_A^y \otimes \proj{y}_Y,
\end{equation}
\begin{equation} \label{eq:theo:main32}
	\sigma_{AY} = \sum_{y \in \sy} s_Y(y) \sigma_A^y \otimes \proj{y}_Y,
\end{equation}
where $r_Y, s_Y \in \mathcal{P}_{\sy}$ and $\left\lbrace \rho_A^y \right\rbrace_{y \in \sy}$ and $\left\lbrace \sigma_A^y \right\rbrace_{y \in \sy}$ are sets of states of $A$.
Suppose without loss of generality that $\pr{A}{Y}{\rho} \leq \pr{A}{Y}{\sigma}$. Define the conditional dephasing channel as
\begin{equation}
	\mathcal{N}_{AY}[\omega_{AY}] = \sum_{x \in \sx, y \in \sy} \proj{\phi^y_x,y} \omega_{AY} \proj{\phi^y_x,y},
\end{equation}
where $\ket{\phi^y_x,y} \equiv \ket{\phi^y_x} \otimes \ket{y}$ and the states $\ket{y}$ and $\ket{\phi^y_x}$ are defined through \eqref{eq:rhoAY2} and \eqref{eq:rhoAY}. Here and henceforth we suppress the subscripts $A$ and $Y$ (denoting the subsystems) for notational simplicity. We have that
\begin{equation}
	\mathcal{N}_{AY}[\rho_{AY}] = \rho_{AY},
\end{equation}
while
\begin{equation}
	\mathcal{N}_{AY}[\sigma_{AY}] = \sum_{y \in \sy, x \in \sx} s_Y(y) \tilde{s}_{X|Y}(x|y) \proj{\phi^y_x} \otimes \proj{y},
\end{equation}
where
\begin{equation}
	\tilde{s}_{X|Y}(x|y) = \bra{\phi^y_x} \sigma_A^y \ket{\phi^y_x}, \quad \forall x \in \sx, \forall y \in \sy,
\end{equation}
so that $\tilde{s}_{X|Y}$ is a conditional probability distribution. Now, for each $y \in \sy$, consider the spectral decomposition of $\sigma_A^y$:
\begin{equation}
	\sigma_A^y = \sum_{x \in \sx} s_{X|Y}(x|y) \proj{\psi^y_x},
\end{equation}
where $|\sx| = d_A$ is the dimension of system $A$, $s_{X|Y}$ is a conditional probability distribution defined as
\begin{equation}
	s_{X|Y}(x|y) = \bra{\psi^y_x} \sigma_A^y \ket{\psi^y_x}, \quad \forall x \in \sx, \forall y \in \sy,
\end{equation}
and $\left\lbrace \psi^y_x \right\rbrace_{x \in \sx}$ is a set of orthonormal states for a fixed value of $y \in \sy$. Since for any positive semi-definite matrix, the vector of diagonal elements in any basis is majorized by the vector of eigenvalues~\cite{Majorization}, we have that
\begin{equation}
	\tilde{s}_{X|Y=y} \prec s_{X|Y=y}, \quad \forall y \in \sy.
\end{equation}
Define
\begin{equation} \begin{aligned}
  s_{XY}(x,y) & = s_Y(y) s_{X|Y}(x|y), \\
  \tilde{s}_{XY}(x,y) & = s_Y(y) \tilde{s}_{X|Y}(x|y), \end{aligned} \qquad \forall x\in \sx, y \in \sy,
\end{equation}
so that
\begin{equation}
	\tilde{s}_Y(y) = \sum_{x \in \sx} \tilde{s}_{XY}(x,y) = s_Y(y) = \sum_{x \in \sx} s_{XY}(x,y),
\end{equation}
for all $y \in \sy$. Using Lemma \ref{lem:cmajToMaj}, we have that
\begin{equation}
	\tilde{s}_{XY} \cmajj s_{XY},
\end{equation}
so that according to Lemma \ref{lem:CondSchurConcavity},
\begin{equation}
	\pr{X}{Y}{s} \leq \pr{X}{Y}{\tilde{s}}, \quad \forall \a \in [0,1).
\end{equation}
Now, using \eqref{eq:PetzToArimoto}, we have
\eqal{
  \pr{X}{Y}{r} & = \pr{A}{Y}{\rho} \\
  & \leq \pr{A}{Y}{\sigma} \\
  & = \pr{X}{Y}{s} \\
  & \leq \pr{X}{Y}{\tilde{s}}.
}
Meanwhile, from data processing inequality for the normalized trace distance,
\eqal{
	\frac{1}{2} || \rho_{AY} - \sigma_{AY} ||_1 & \geq \frac{1}{2} \left\Vert \mathcal{N}_{AY}[\rho_{AY}] - \mathcal{N}_{AY}[\sigma_{AY}] \right\Vert_1 \\
	& = \mathrm{TV}(r_{XY},\tilde{s}_{XY}),
}
so that
\begin{equation}
	\mathrm{TV}(r_{XY},\tilde{s}_{XY}) \leq \epsilon.
\end{equation}
Using Theorem \ref{theo:main}, we have
\eqal{
  & |\pr{A}{Y}{\rho} - \pr{A}{Y}{\sigma}| \\
  & = \pr{A}{Y}{\sigma} - \pr{A}{Y}{\rho} \\
	& \leq \pr{X}{Y}{\tilde{s}} - \pr{X}{Y}{r} \\
	& \leq \frac{1}{1-\a} \log \left( \left( 1- \epsilon \right)^{\a} + \left(|\sx|-1\right)^{1-\a} \epsilon^{\a} \right).
}
In order to show that the inequality is tight, one can simply consider states that are diagonal in the same basis and whose matrix elements are given by \eqref{eq:qtight} and \eqref{eq:ptight}.

\section{Discussion and open problems}

In this paper, we have proven tight uniform continuity bounds for the $\a$-Arimoto-R\'enyi conditional entropy for $\alpha \in [0,1)$, for joint probability distributions of a pair of discrete random variables with finite alphabets, as well as for the conditional R\'enyi entropy for classical-quantum systems, with the conditioning being on the classical system. In the limit $\a \to 1$, our results yield the corresponding recently obtained bounds for the conditional Shannon entropy and the conditional entropy of a c-q state, respectively. It is interesting to note that our bound for the $\a$-ARCE is identical to the one obtained by Audenaert~\cite{Audenaert} for the unconditional R\'enyi entropy (see (7) of the Appendix of~\cite{Audenaert}). The same is true for the bound obtained by Alhejji and Smith~\cite{AlhejjiSmith} for the conditional Shannon entropy: it is identical to the bound for the Shannon entropy (see~\cite{Zhang}). It would be interesting to see whether an intuitive reason for this can be found. 
\smallskip

A natural next step would be to investigate the continuity of the $\a$-conditional R\'enyi entropy of bipartite quantum systems or classical-quantum systems with the conditioning being on the quantum system, for which there are no results even for the quantum conditional entropy. Finding analogues of our results for $\a >1$ also remains open.


%

\appendices

\section{Conditional majorization \label{CondMajGour}}

The following notion of conditional majorization was introduced by Gour {\em{et al}} in~\cite{CondMaj}.

\begin{defi} \label{defi:condMajGour}
Denote by $\mathbb{R}^{n \times l}_+$ the set of all $n \times l$ matrices with non-negative values. Consider $P \in \mathbb{R}^{n \times l}_+$ and $Q \in \mathbb{R}^{n \times m}_+$. We say $Q$ is conditionally majorized by $P$, written $Q \gmajj P$, if there exist matrices $D^{(j)}$ and $R^{(j)}$, where $j$ can run over an arbitrary number of values, such that
\begin{equation}
	Q = \sum_j D^{(j)} P R^{(j)},
\end{equation}
where each $D^{(j)}$ is an $n \times n$ doubly-stochastic matrix and each $R^{(j)}$ is an $l \times m$ matrix of non-negative entries, with $\sum_j R^{(j)}$ row-stochastic.
\end{defi}


The following necessary and sufficient condition for conditional majorization was proven by Gour {\em{et al}} ( see Theorem 1 of~\cite{CondMaj}):

\begin{lem} \label{lem:CharacCondMajGour}
Let $p_{XY}, q_{XY} \in \mathcal{P}_{\sx \times \sy}$. We have that $q_{XY} \gmajj p_{XY}$ if and only if for all convex symmetric functions $\Phi$ 
\begin{equation}
    \sum_{y \in \sy} p_Y(y) \Phi\left( p_{X|Y=y} \right) \geq \sum_{y \in \sy} q_Y(y) \Phi\left( q_{X|Y=y} \right).
\end{equation}
\end{lem}


\section{Further tools for the proofs}

\begin{lem} \label{lem:majEps}
Consider a vector $\bos{v} \in \mathbb{R}^n$ whose elements are non-negative and are arranged in non-increasing order. For some $i \in \left\lbrace 2, \cdots, n \right\rbrace$ and some $s \in (0, v_i]$, define the vector $\bos{v}^{(i)} \in \mathbb{R}^n$ whose elements satisfy the following relations:
\eqs{
	v^{(i)}_1 & = v_1 + s, \\
	v^{(i)}_i & = v_i - s, \\
	v^{(i)}_j & = v_j, \quad \forall j \in \left\lbrace 2, \cdots, n \right\rbrace \setminus \left\lbrace i \right\rbrace.
}
Then $\bos{v} \prec \bos{v}^{(i)}$.
\end{lem}
\begin{proof}
We have that $v^{(i)}_j \geq v^{(i)}_{j+1}$ for all $j \in \left\lbrace 1, \cdots, i-1 \right\rbrace$. Now, there exists $k \in \left\lbrace i+1, \cdots, n \right\rbrace$ such that $v^{(i)}_i \leq v^{(i)}_j$ for all $j \in \left\lbrace 1, \cdots, k \right\rbrace$ and $v^{(i)}_i > v^{(i)}_j$ for all $j \in \left\lbrace k+1, \cdots, n \right\rbrace$, so that the vector $\bos{v}^{(i)\downarrow} \coloneqq \left( v^{(i)}_1, \cdots, v^{(i)}_{i-1}, \cdots, v^{(i)}_k, v^{(i)}_i, v^{(i)}_{k+1}, \cdots, v^{(i)}_n \right)^{\mathrm{T}}$ has its elements arranged in non-increasing order. In that case,
\begin{equation}
	\sum_{j=1}^{l} v^{(i)\downarrow}_j = \sum_{j=1}^{l} v_j + s > \sum_{j=1}^{l} v_j, \quad \forall l = 1, \cdots i-1.
\end{equation}
Using the equality in the above equation, we have for all $l = i, \cdots k-1$,
\eqal{ \label{eq:lem:majEps}
	\sum_{j=1}^{l} v^{(i)\downarrow}_j & = \sum_{j=1}^{i-1} v_j + s + \sum_{j=i}^{l} v^{(i)\downarrow}_j \\
	& = \sum_{j=1}^{i-1} v_j + s + \sum_{j=i+1}^{l+1} v_j \\
	& = \sum_{j=1}^l v_j + (s - v_i) + v_{l+1} \\
	& > \sum_{j=1}^l v_j.
}
In order to obtain the last inequality in the above set of equations, notice that $v^{(i)}_i \leq v^{(i)}_k$, or, $v_i-s \leq v_k$. Now, $l+1 \leq k$, and since the elements of $\bos{v}$ are sorted in non-increasing order, $v_{l+1} \geq v_k$, so that $v_i-s \leq v_{l+1}$, which implies the last inequality in \eqref{eq:lem:majEps}. Similarly, using the last equality in the above equation,
\eqal{
	\sum_{j=1}^{k} v^{(i)\downarrow}_j & = \sum_{j=1}^{k-1} v^{(i)\downarrow}_j + v^{(i)\downarrow}_k \\
	& = \sum_{j=1}^{k-1} v_j + (s - v_i) + v_k + (v_i - s) \\
	& = \sum_{j=1}^{k} v_j.
}
Finally, using the above equation, we have for all $l = k+1, \cdots n$,
\eqal{
	\sum_{j=1}^{l} v^{(i)\downarrow}_j & = \sum_{j=1}^{k} v^{(i)\downarrow}_j + \sum_{j=k+1}^{l} v^{(i)\downarrow}_j \\
	& = \sum_{j=1}^{k} v_j + \sum_{j=k+1}^{l} v_j \\
	& = \sum_{j=1}^{l} v_j.
}
This ends the proof.
\end{proof}


\begin{lem} \label{lem:MajPerp}
Let $\bos{v}, \bos{v}', \bos{v}^{\perp} \in \mathbb{R}^n$ be three vectors with non-negative entries, such that $\bos{v}$ and $\bos{v}'$ are in the same subspace whereas $\bos{v}^{\perp}$ is in an orthogonal subspace. If $\bos{v}' \prec \bos{v}$, then $\bos{v}'+\bos{v}^{\perp} \prec \bos{v}+\bos{v}^{\perp}$.
\end{lem}
\begin{proof}
If $\bos{v}' \prec \bos{v}$, there exist $\mu_j \in \mathbb{R}$ such that $\mu_j \geq 0, \forall j$ and $\sum_j \mu_j = 1$ and permutation matrices $\Pi_j$ of dimension $n$ such that
\begin{equation}
	\bos{v}' = \sum_j \mu_j \Pi_j \bos{v}.
\end{equation}
The permutations $\Pi_j$ act only on the subspace spanning $\bos{v}$ and $\bos{v}'$. As a consequence,
\begin{equation}
\bos{v}'+\bos{v}^{\perp} = \sum_j \mu_j \Pi_j \left( \bos{v}+\bos{v}^{\perp} \right),
\end{equation}
so that $\bos{v}'+\bos{v}^{\perp} \prec \bos{v}+\bos{v}^{\perp}$.
\end{proof}


\begin{lem} \label{lem:majEpsComp}
Consider a vector $\bos{v} \in \mathbb{R}^{n+m}$ whose elements are non-negative and satisfy $v_j = 0$ for all $j \in \left\lbrace n+1, \cdots, n+m \right\rbrace$. For some $i \in \left\lbrace 1, \cdots, n \right\rbrace$, some $j \in \left\lbrace n+1, \cdots, n+m \right\rbrace$ and some $s \in (0, v_i]$, define the vector $\bos{u} \in \mathbb{R}^n$ whose elements satisfy the following relations:
\eqs{
	u_i & = v_i - s, \\
	u_j & = s, \\
	u_k & = v_k, \quad \forall k \in \left\lbrace 1, \cdots, n+m \right\rbrace \setminus \left\lbrace i,j \right\rbrace.
}
Then $\bos{u} \prec \bos{v}$.
\end{lem}
\begin{proof}
Define the vectors $\bos{\tilde{v}}$, $\bos{\tilde{u}}$ and $\bos{v}^{\perp} \in \mathbb{R}^{n+m}$ whose elements are non-negative and satisfy the following relations:
\eqs{
	\tilde{v}_i & = v_i, \\
	\tilde{v}_k & = 0, \quad \forall k \in \left\lbrace 1, \cdots, n+m \right\rbrace \setminus \left\lbrace i \right\rbrace,
}
\eqs{
	\tilde{u}_i & = v_i - s, \\
	\tilde{u}_j & = s, \\
	\tilde{u}_k & = 0, \quad \forall k \in \left\lbrace 1, \cdots, n+m \right\rbrace \setminus \left\lbrace i,j \right\rbrace,
}
and
\eqs{
	v^{\perp}_i & = 0, \\
	v^{\perp}_j & = 0, \\
	v^{\perp}_k & = v_k, \quad \forall k \in \left\lbrace 1, \cdots, n+m \right\rbrace \setminus \left\lbrace i,j \right\rbrace.
}
Note that $\bos{v} = \bos{\tilde{v}} + \bos{v}^{\perp}$ and $\bos{u} = \bos{\tilde{u}} + \bos{v}^{\perp}$. Now, we trivially have that $\bos{\tilde{u}} \prec \bos{\tilde{v}}$, so that, using Lemma \ref{lem:MajPerp}, $\bos{u} \prec \bos{v}$.
\end{proof}


\begin{lem} \label{lem:IncDiffLog}
Let $\a \in [0,1)$, $|\sx| \geq 1$ and $\tpp \in (0, 1-\frac{1}{|\sx|}]$. The function $f_{\a,|\sx|,\tpp} : [0, \infty) \rightarrow [0, \infty)$ defined as
\begin{equation}
	f_{\a,|\sx|,\tpp}(u) \coloneqq \log \left[ u^\a + \left( |\sx|-1 \right)^{1-\a} \tpp^\a \right] - \log \left[ \left( u + \tpp \right)^\a \right]
\end{equation}
is monotonically decreasing for $u \geq 1-\tpp$.
\end{lem}
\begin{proof}
We simply compute the derivative
\eqaln{
  & \frac{\partial}{\partial u} f_{\a,|\sx|,\tpp}(u) \\
  & = \frac{\partial}{\partial u} \left( \log \left[ u^\a + \left( |\sx|-1 \right)^{1-\a} \tpp^\a \right] - \log \left[ \left( u + \tpp \right)^\a \right] \right) \\
	& = \left[ u^\a + \left( |\sx|-1 \right)^{1-\a} \tpp^\a \right]^{-1} \a u^{\a-1} - \alpha(u + \tpp)^{-1}
}
which will be $\leq 0$ if and only if
\begin{align}
    \frac{u^{\a-1}}{u^\a + \left( |\sx|-1 \right)^{1-\a} \tpp^\a}
    & \leq \frac{1}{u + \tpp},
\end{align}
which is easily seen to reduce to the condition
\begin{align}\label{cond}
   u \geq \tpp/(|\sx|-1)
\end{align}
Note that since $\tpp \in (0, 1-\frac{1}{|\sx|}]$, we have $1-\tpp\geq \tpp/(|\sx| -1)$. Hence for $u \geq 1 - \tpp$ the required condition \eqref{cond} holds.

\end{proof}


\begin{lem} \label{lem:BoundMonotonic}
Let $\a \in [0,1)$ and $|\sx| \geq 1$. The function $g_{\a,|\sx|} : (0, 1-\frac{1}{|\sx|}] \rightarrow [0, \infty)$ defined as 
\begin{equation}
	g_{\a,|\sx|}(u) \coloneqq \left( 1 - u \right)^\a + \left( |\sx|-1 \right)^{1-\a} u^\a
\end{equation}
is monotonically increasing in $u$.
\end{lem}
\begin{proof}
Again, we simply compute the derivative
\eqaln{
	\frac{\partial}{\partial u} g_{\a,|\sx|}(u) & = \frac{\partial}{\partial u} \left[ \left( 1 - u \right)^\a + \left( |\sx|-1 \right)^{1-\a} u^\a \right] \\
	& = - \a \left( 1 - u \right)^{\a-1} + \left( |\sx|-1 \right)^{1-\a} \a u^{\a-1}
}
which is $\geq 0$ if and only if the following inequality holds:
\begin{align}\label{cond2}
	\left( |\sx|-1 \right)^{1-\a} u^{\a-1} & \geq \left( 1 - u \right)^{\a-1}.
	\end{align}
	The above inequality can easily be seen to reduce to the following:
	\begin{align}
	    \left( (|\sx|-1) (1-u)/u \right)^{1-\a} & \geq 1,
	\end{align}
	which in turn reduces to the condition $u \leq 1 - \frac{1}{|\sx|}$, which holds by the hypothesis of the lemma.

\end{proof}

\section{Illustration of the steps in the proof of Theorem~\ref{theo:main}  \label{exampleARCEproof}}

In this section of the appendix, we consider a simple $3 \times 3$ example in order to illustrate the steps in the proof of Theorem~\ref{theo:main}. In the context of this example, we relabel the matrix $q_{XY}$ simply as $q$ and its elements as $q_{ij}$. Similarly, we relabel the matrix $p_{XY}$ simply as $p$ and its elements as $p_{ij}$. The example is illustrated in the three figures given below.
\begin{figure}[ht]
\centering
	\includegraphics[width=.9\columnwidth]{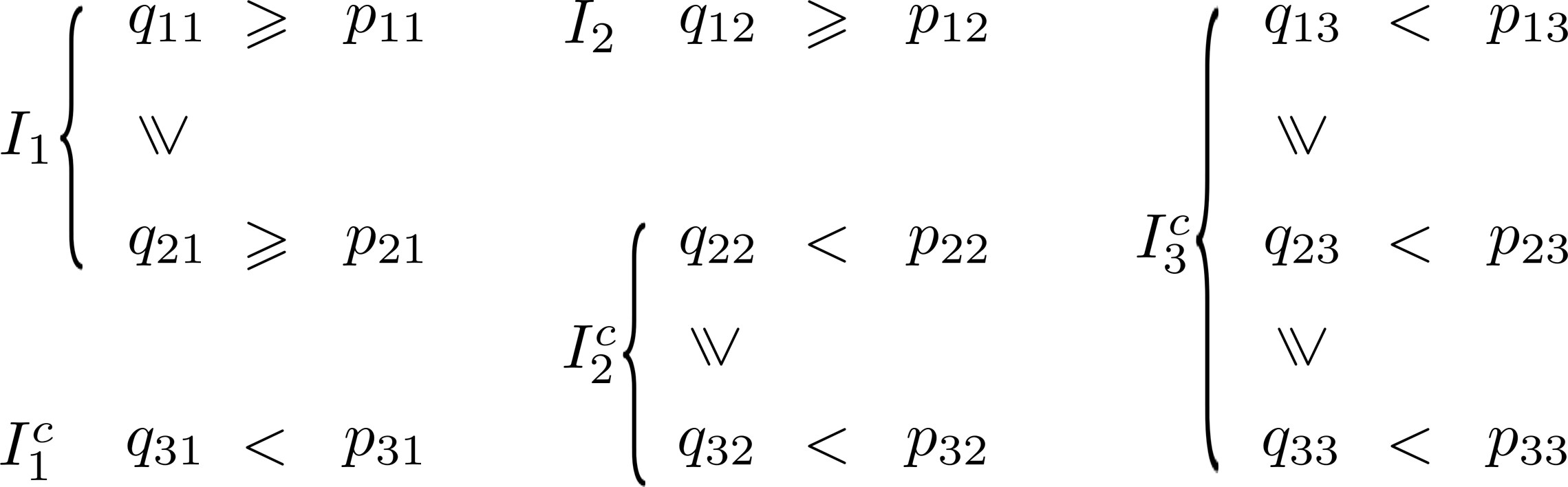}
	\caption{\label{exampleARCEproofA} Step A of the proof of Theorem~\ref{theo:main}, which corresponds to a reordering of some elements of $q$ and $p$.}
\end{figure}
\begin{figure}[ht]
\centering
	\includegraphics[width=.9\columnwidth]{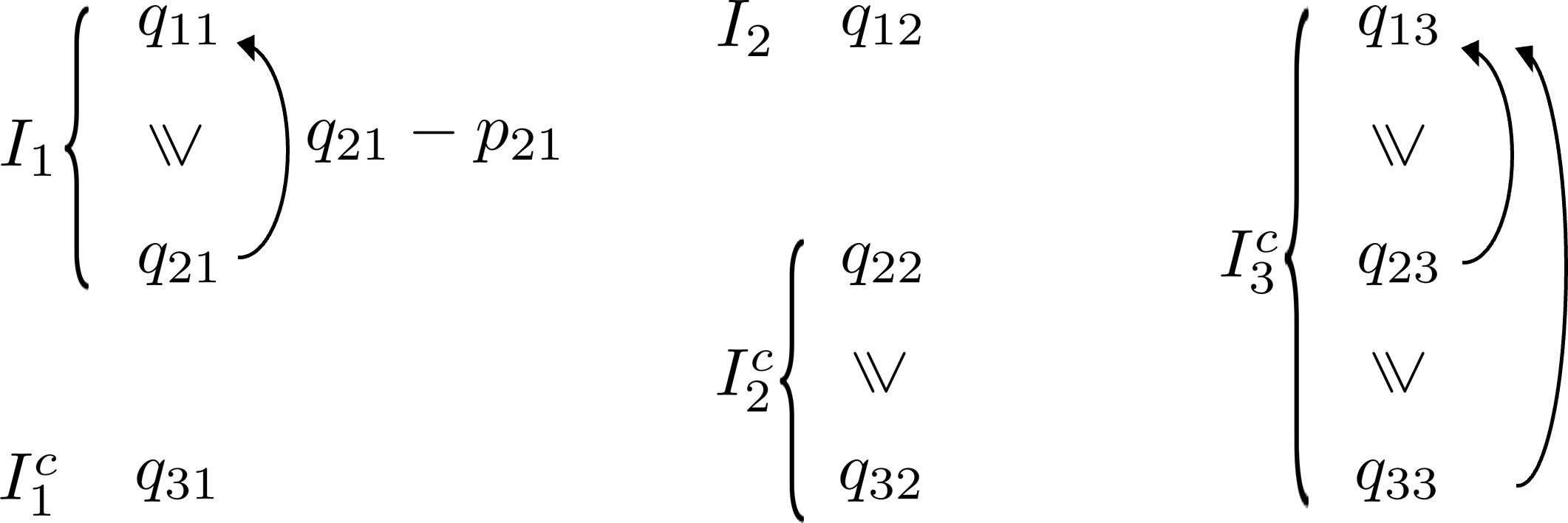}
	\caption{\label{exampleARCEproofB} Step B of the proof of Theorem~\ref{theo:main}, in which some probability weights are moved in $q$, while $p$ remains unchanged. In this example, we consider that after the moves, we end up with a new matrix $q$ that satisfies $q_{13} < p_{13}$, $q_{23} = 0$ and $q_{33} = 0$. This will affect Step C, see Figure~\ref{exampleARCEproofC}.}
\end{figure}
\begin{figure}[ht!]
\centering
	\includegraphics[width=.9\columnwidth]{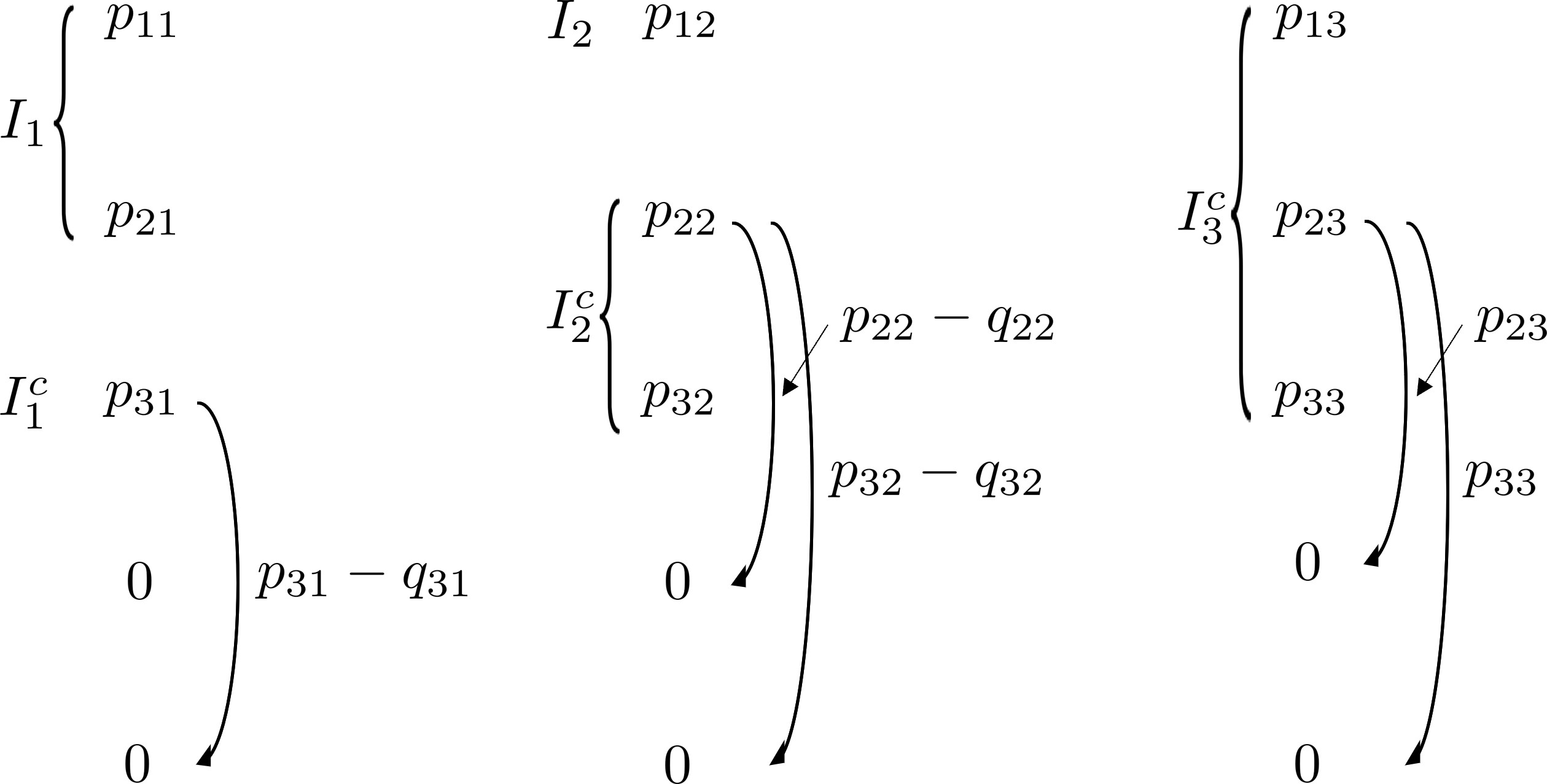}
	\caption{\label{exampleARCEproofC} Step C of the proof of Theorem~\ref{theo:main}, in which we increase the dimensions of $q$ and $p$, and some probability weights are moved in $p$.}
\end{figure}

\section*{Acknowledgment}

M. G. J. acknowledges support from the Wiener-Anspach Foundation.




\bibliographystyle{IEEEtran}
\bibliography{ARCEbound_IEEE}

\end{document}